\documentclass{amsart}

\usepackage{hyperref}
\usepackage{amsmath}
\usepackage{amssymb}
\usepackage{graphicx}
\usepackage{mathdots}
\usepackage{gensymb}
\usepackage{epstopdf}

\newtheorem{thm}{Theorem}%[section]
\newtheorem{prop}[thm]{Proposition}

\newtheorem{cor}[thm]{Corollary}

\theoremstyle{definition}

\theoremstyle{remark}

\numberwithin{equation}{section}

\newcommand{\hur}{H}

\newcommand{\mon}{\vec{H}}

\newcommand{\Tr}{\operatorname{Tr}}

\newcommand{\diag}{\operatorname{diag}}

\begin{document}

%%
%% The title of the paper goes here.  Edit to your title.
%%

\title[Lozenge Tilings and Hurwitz Numbers]{Lozenge Tilings
and Hurwitz Numbers}
%\keywords{Random tilings, random matrices.}
%\subjclass[2010]{test}
%%
%% Now edit the following to give your name and address:
%% 

\author{Jonathan Novak}
\address{Department of Mathematics, Massachusetts Institute of Technology, USA}
\email{jnovak@math.mit.edu}

%%%
%%% The following is for the abstract.  The abstract is optional and
%%% if not used just delete, or comment out, the following.
%%%

%%
%%  LaTeX will not make the title for the paper unless told to do so.
%%  This is done by uncommenting the following.
%%

 \maketitle
%\tableofcontents
%%
%% LaTeX can automatically make a table of contents.  This is done by
%% uncommenting the following:
%%

%%
%%  To enter text is easy.  Just type it.  A blank line starts a new
%%  paragraph. 
%%

\begin{abstract}
    We give a new proof of the fact that,
    near a turning point of the frozen 
    boundary, the vertical tiles in a
    uniformly random lozenge tiling of a 
    large sawtooth domain are distributed like the eigenvalues
    of a GUE random matrix.  Our argument 
    uses none of the standard tools of integrable probability.
    In their place, it uses a combinatorial interpretation
    of the Harish-Chandra/Itzykson-Zuber integral
    as a generating function for desymmetrized
    Hurwitz numbers.
\end{abstract}

\section{Introduction}

\subsection{}
Let

	\begin{equation*}
		\label{eqn:TriangularArray}
			\begin{matrix}
				b_1^{(1)} & {} & {} & {} \\
				b_1^{(2)} & b_2^{(2)} & {} & {} \\
				b_1^{(3)} & b_2^{(3)} & {b_3^{(3)}} & {} \\
				\vdots & \vdots & \vdots & \ddots
			\end{matrix}
	\end{equation*} 
	   
\noindent
be a triangular array of integers,
the elements of which are strictly 
decreasing along rows.  The array
\eqref{eqn:TriangularArray} gives
rise to a sequence $\Omega^{(N)}$ of planar 
domains via the following construction.
Fix a coordinate system in the plane whose axes meet at a 
$120^{\degree}$ angle.  
We specify $\Omega^{(N)}$ by specifying its boundary,
which consists of two piecewise linear components.  
One component of
$\partial\Omega^{(N)}$ --- the lower boundary --- is simply the
horizontal axis in the plane.  The other component --- the 
upper boundary --- is built in three steps.  First,
construct the line parallel to the lower boundary passing 
through the point $(0,N)$.  Second, affix $N$ outward-facing 
unit triangles to this line such that the midpoints of their 
bases have horizontal coordinates 
$b_1^{(N)} > \dots > b_{N}^{(N)}$.
Finally, erase the bases of these triangles.  We will refer to
$\Omega^{(N)}$ as the \emph{sawtooth domain} of rank $N$ with
boundary conditions $(b_1^{(N)},\dots,b_{N}^{(N)})$.

\begin{figure}
	\includegraphics{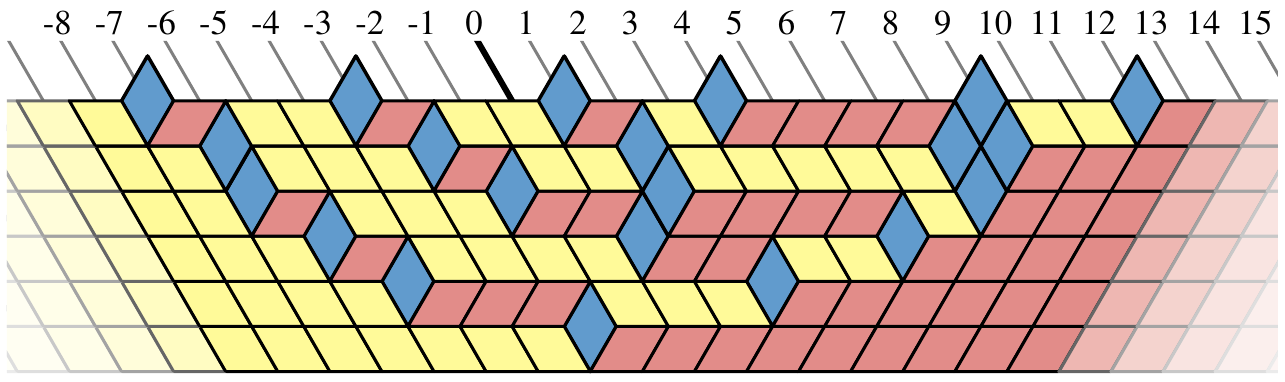}
	\caption{\label{fig:Tiling}A lozenge tiling of
	a sawtooth domain of rank $6$.}
\end{figure}

\subsection{}
A \emph{lozenge} is a unit rhombus in the plane whose
sides are parallel to one of the coordinate axes, or to the line
bisecting the obtuse angle between them.  Lozenges
are thus divided into three classes: left-leaning, right-leaning,
and vertical.  Given a lozenge tiling of $\Omega^{(N)}$,
as in Figure \ref{fig:Tiling}, the horizontal
line through $(0,k)$ ``threads'' exactly 
$k$ vertical tiles, or ``beads'',
and the beads on adjacent 
threads interlace, as in Figure
\ref{fig:Beads}.

Let $T^{(N)}$ be a
uniformly random lozenge tiling of $\Omega^{(N)}$,
and let $b_{k1}^{(N)} > \dots > b_{kk}^{(N)}$
be the horizontal coordinates of the
centroids of the beads on the
$k$th thread through $T^{(N)}$.
The main result of this note is a limit theorem
for the $k$-dimensional
random vector $(b_{k1}^{(N)},\dots,b_{kk}^{(N)})$,
in the regime where $N \rightarrow \infty$ 
with $k$ fixed.

\subsection{}
Suppose there exists a positive integer $M$ such
that, for each $N \geq 1$,

\begin{equation*}
    \{b_1^{(N)} > \dots > b_N^{(N)}\} \subseteq
    \{MN > \dots > -MN\}.
\end{equation*}

\noindent
Let $\nu^{(N)}$ be the probability measure which
places mass $1/N$ at each of the points $b_i^{(N)}/N$.
Suppose that $\nu^{(N)}$ converges weakly to $\nu$, the 
probability measure on $[-M,M]$ with moment 
sequence $\psi_1,\psi_2,\psi_3,\dots$.

\begin{thm}
	\label{thm:Main}
	For each $N \geq 1$ and $1\leq k\leq N$, set 
	
		\begin{equation*}
			\tilde{b}_{kl}^{(N)} = 
			\frac{\frac{b_{kl}^{(N)}}{\sqrt{N}}-
			(\psi_1-\frac{1}{2})\sqrt{N}}
			{\psi_2-\psi_1^2-\frac{1}{12}}, \quad 
			1 \leq l \leq k.
		\end{equation*}
		
	\noindent
	For any fixed $k$, the random vector 
	$(\tilde{b}_{k1}^{(N)},\dots,\tilde{b}_{kk}^{(N)})$
	converges weakly to the ordered list of 
	eigenvalues of a $k \times k$
	GUE random matrix as $N \rightarrow \infty$.
\end{thm}

Note that $\psi_1$ and $\psi_2-\psi_1^2$ are,
respectively, the mean and variance of $\nu$,
while the numbers $1/2$ and $1/12$ are
the mean and variance of the uniform 
probability measure on $[0,1]$.  

\subsection{}
Given that the law of large numbers for $T^{(N)}$
manifests as the convergence 
of the height function of the normalized tiling 
$N^{-1}T^{(N)}$ to a deterministic 
limit, the so-called \emph{limit shape} \cite{DM,KO,Petrov},
the $N^{-1/2}$ scaling
in Theorem
\ref{thm:Main} is natural.
Indeed, as discussed in
\cite{OR}, the arctic curve separating the
frozen and liquid regions of $T^{(N)}$ which
emerge as $N \rightarrow \infty$
resembles a parabola near the point where
it is tangent to the lower boundary of
$\Omega^{(N)}$.  For boundary conditions producing
an arctic curve which actually is a parabola,
see \cite{NY1,NY2}.

\begin{figure}
	\includegraphics{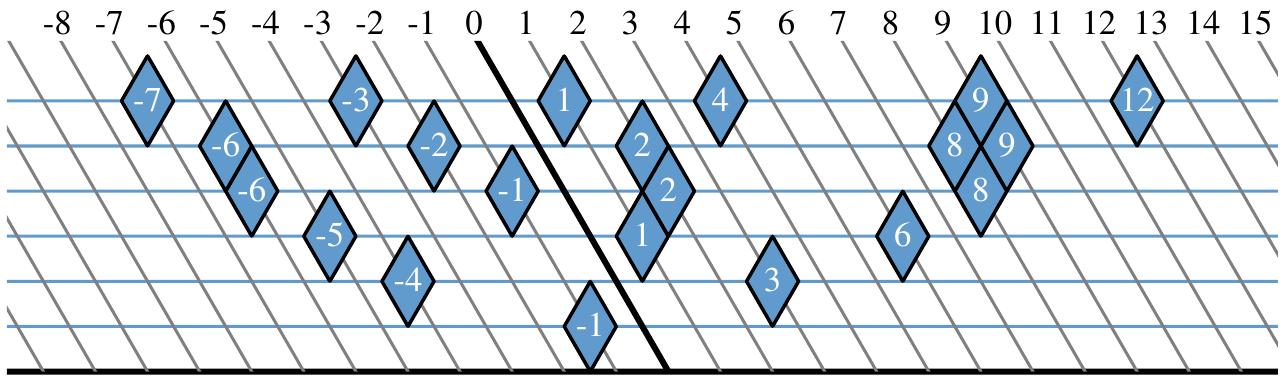}
	\caption{\label{fig:Beads}Interlacing beads and their coordinates.}
\end{figure}

\subsection{}
The connection between the joint distribution
of vertical tiles near the frozen boundary and
GUE eigenvalues was first studied by 
by Okounkov and Reshetikhin \cite{OR}.  
For a special class of boundary conditions,
Theorem \ref{thm:Main} was proved
by Johansson and Nordenstam \cite{JN}.  
In a slightly different (but equivalent) form, 
Theorem \ref{thm:Main} was obtained in full
generality by Gorin and Panova
\cite{GP} as a consequence of their general approach to 
Schur function asymptotics.  
In this note, we present a different approach
to Theorem \ref{thm:Main} in which the 
usual tools of integrable probability (e.g. determinantal processes,
steepest descent analysis) play no role.  Instead,
our argument is based on the combinatorial
interpretation of the Harish-Chandra/Itzykson-Zuber
integral discovered in \cite{GGN3}.

\subsection{}
Work on this paper began while the author was
a Professeur Invit\'e at
Universit\' e Paris Diderot in
the Spring of 2014.  I am grateful to
G. Chapuy and S. Corteel for
the invitation to visit.  
While writing this article, I benefited
from stimulating correspondence 
with V. Gorin and G. Panova.
I am indebted to
M. Lacroix for producing the figures which
accompany this note.

\section{Proof of Theorem \ref{thm:Main}}

\subsection{}
Let us replace the 
$k$-dimensional random vector 
$(b_{k1}^{(N)},\dots,b_{kk}^{(N)})$
with the random Hermitian matrix

	\begin{equation*}
		B_k^{(N)} = U_k
		\begin{bmatrix}
			b_{k1}^{(N)} & {} & {} \\
			{} & \ddots & {} \\
			{} & {} & b_{kk}^{(N)}
		\end{bmatrix}
		U_k^{-1},
	\end{equation*}

\noindent
where $U_k$ is a random matrix drawn from 
normalized Haar measure on the
unitary group $U(k)$.  By the \emph{Laplace transform} of 
$B_k^{(N)}$, we mean the function on $k \times k$ 
complex semisimple matrices $A$ defined by 

	\begin{equation*}
	\label{eqn:Laplace}
		A \mapsto \mathbf{E}[e^{\Tr AB_k^{(N)}}],
	\end{equation*}
	
\noindent
where $\mathbf{E}$ denotes expectation.  In the 
case $k=1$, this function coincides with the classical
two-sided Laplace transform encoding the distribution
of the horizontal coordinate of the bottom bead.

\subsection{}
The Laplace transform of $B_k^{(N)}$
depends only on the 
eigenvalues of $A$, and 
thus may be considered as a function
of $k$ complex variables.  This function
is analytic, because the distribution
of $B_k^{(N)}$ in $H(k)$, the space of
$k \times k$ Hermitian matrices,
is compactly supported.
Explicitly,

	\begin{equation*}
	\begin{split}
		L_k^{(N)}(a_1,\dots,a_k) &= 
		\sum_{\{b_1 > \dots > b_k\} \subset \mathbb{Z}}
		\mathbf{P}(b_{k1}^{(N)}=b_1,\dots,b_{kk}^{(N)}=b_k) \\
		&\times\int\limits_{U(k)} e^{\Tr \diag(a_1,\dots,a_k)U
		\diag(b_1,\dots,b_k)U^{-1}} \mathrm{d}U,
	\end{split}
	\end{equation*}

\noindent
where the sum is over all $k$-point
particle configurations on the integer
lattice and $\mathbf{P}$ is the uniform 
probability measure on lozenge 
tilings of $\Omega^{(N)}$.  The integral
over $U(k)$ is just
the Laplace transform of the uniform 
probability measure
on the set of $k \times k$ Hermitian
matrices with eigenvalues 
$b_1 > \dots > b_k$.  That is,
$L_k^{(N)}$ is
the Laplace transform of a mixture of
\emph{orbital measures}.
If $k=N$, the bead locations are deterministic,
and we are dealing with the Laplace transform of a 
pure orbital measure.  The following proposition
reduces our workload to the analysis of the 
Laplace transforms of pure orbital measures.

\begin{prop}
\label{prop:Key}
For any integers
$1 \leq k \leq N$,

	\begin{equation*}
		L_k^{(N)}(a_1,\dots,a_k) = 
		\left( \prod_{i=1}^k \frac{a_i}{e^{a_i}-1} 
		\right)^{N-k} L_N^{(N)}(a_1,\dots,a_k,0,\dots,0).
	\end{equation*}
	
\end{prop}
	
\begin{proof}
    The proof is a combination of
    three standard facts from
    the representation theory of the complex
    general linear group $GL(N)$.
    
    First, the isomorphism classes
    of irreducible rational representations
    of $GL(N)$ are indexed by $N$-point particle
    configurations on $\mathbb{Z}$.
    This is a classical result,
    see e.g. \cite{Weyl}.
    
    Second,
    given a particle configuration
    $\{b_1 > \dots > b_N\} \subset \mathbb{Z}$, 
    the corresponding normalized irreducible character
    
    \begin{equation*}
        \frac{\chi^{(b_1,\dots,b_N)}(e^{a_1},\dots,e^{a_N})}
        {\chi^{(b_1,\dots,b_N)}(1,\dots,1)}
    \end{equation*}

    \noindent
    equals the twisted Laplace transform
    
    \begin{equation*}
        \prod_{1\leq i<j \leq N} \frac{a_i-a_j}{e^{a_i}-e^{a_j}} 
        \int\limits_{U(N)}
        e^{\Tr \diag(a_1,\dots,a_N)U\diag(b_1,\dots,b_N)U^{-1}}
        \mathrm{d}U
    \end{equation*}
    
    \noindent
    of the uniform measure on Hermitian 
    matrices with spectrum $\{b_1 > \dots > b_N\}$.
    This identity is independently due to Harish-Chandra 
    \cite{HC}, and Itzykson and Zuber \cite{IZ} 
    --- it is the $U(N)$ case of the Kirillov
    character formula \cite{Kirillov}.
    
    The third and final ingredient is
    the branching rule for irreducible
    characters of $GL(N)$ under restriction
    to $GL(N-1)$:
    
    \begin{equation*}
        \chi^{(b_1,\dots,b_N)}(e^{a_1},\dots,e^{a_{N-1}},1)
        =\sum_{\{c_1 > \dots > c_{N-1}\}\subset \mathbb{Z}} 
        \chi^{(c_1,\dots,c_{N-1})}(e^{a_1},\dots,e^{a_{N-1}}),
    \end{equation*}

    \noindent
    where the sum is over all configurations of
    $N-1$ particles on $\mathbb{Z}$ which interlace
    with the configuration $\{b_1 > \dots > b_N\}$.
    A proof of the branching rule may be found in
    \cite[Chapter 8]{GW}.  Iterating the branching rule
    $N-k$ times and applying the Harish-Chandra formula
    yields the stated formula for $L_k^{(N)}$ in 
    terms of $L_N^{(N)}$.

\end{proof}

\subsection{}
Consider the analytic function 
$\mathbb{C} \times \mathbb{C}^N \times \mathbb{C}^N 
\rightarrow \mathbb{C}$
defined by

\begin{equation*}
\label{eqn:HCIZ}
    (z;a_1,\dots,a_N;b_1,\dots,b_N) \mapsto
    \int\limits_{U(N)} e^{z\Tr 
    \diag(a_1,\dots,a_N)U\diag(b_1,\dots,b_N)U^{-1}}
        \mathrm{d}U.
\end{equation*}

\noindent
This is the famous Harish-Chandra/Itzykson-Zuber 
integral.  The parameter $z$ may
be called the \emph{coupling constant}, as a 
reference to its origin in the spectral analysis
of coupled random semisimple
matrices with $AB$-interaction
\cite{AvM,IZ}.

The HCIZ integral enjoys
a natural $S(N) \times S(N)$ symmetry: it
is invariant under permutation of the $a$'s
amongst themselves, and the $b$'s amongst 
themselves.
Combining this symmetry
with the fact that the Newton power-sums
form a linear basis of the algebra of
symmetric polynomials, we may present the Maclaurin 
series of the logarithm of the HCIZ 
integral in the form

\begin{equation*}
\begin{split}
    &\log \int\limits_{U(N)} e^{z\Tr 
    \diag(a_1,\dots,a_N)U\diag(b_1,\dots,b_N)U^{-1}}
        \mathrm{d}U \\
    =&\sum_{d=1}^{\infty} \frac{z^d}{d!}
    \sum_{\alpha,\beta \vdash d} C_N(\alpha,\beta)
    p_{\alpha}(a_1,\dots,a_N)p_{\beta}(b_1,\dots,b_N),
\end{split}
\end{equation*}

\noindent
where the internal sum is over all pairs
of Young diagrams with $d$ cells.

\subsection{}
The coefficients $C_N(\alpha,\beta)$
have the following combinatorial 
interpretation.  Consider the
Cayley graph of the symmetric group $S(d)$ as generated
by the conjugacy class of transpositions. 
Equip this graph with the Biane-Stanley
edge labelling \cite{Biane,Stanley}, wherein each
edge corresponding to 
the transposition $(s\ t)$ is tagged
with $t$, the larger of the
two numbers interchanged.  The
$d=4$ case is
shown in Figure \ref{fig:Cayley}, where $2$-edges
are drawn in blue, $3$-edges in yellow, and $4$-edges
in red.
A walk on the Cayley graph is said to be
\emph{monotone} if the labels of the edges it traverses
form a weakly increasing sequence.
A walk is \emph{transitive} if its steps
and endpoints together generate a transitive
subgroup of $S(d)$.  Given
two partitions $\alpha,\beta \vdash d$,
and a nonnegative integer $r$, let $\mon^r(\alpha,\beta)$
be the number of $r$-step monotone,
transitive walks on $S(d)$ which
begin at a permutation of cycle type $\alpha$ and
end at a permutation of cycle type $\beta$.

\begin{thm}[\cite{GGN3}]
\label{thm:LeadingDerivatives}
    For any $1 \leq d \leq N$, and
    any $\alpha,\beta \vdash d$, 
    we have
    
    \begin{equation*}
    C_N(\alpha,\beta) = \frac{1}{N^d}
    \sum_{r=0}^{\infty} (-1)^r
    \frac{\mon^r(\alpha,\beta)}{N^r}.
    \end{equation*}
\end{thm}

\begin{figure}
	\includegraphics{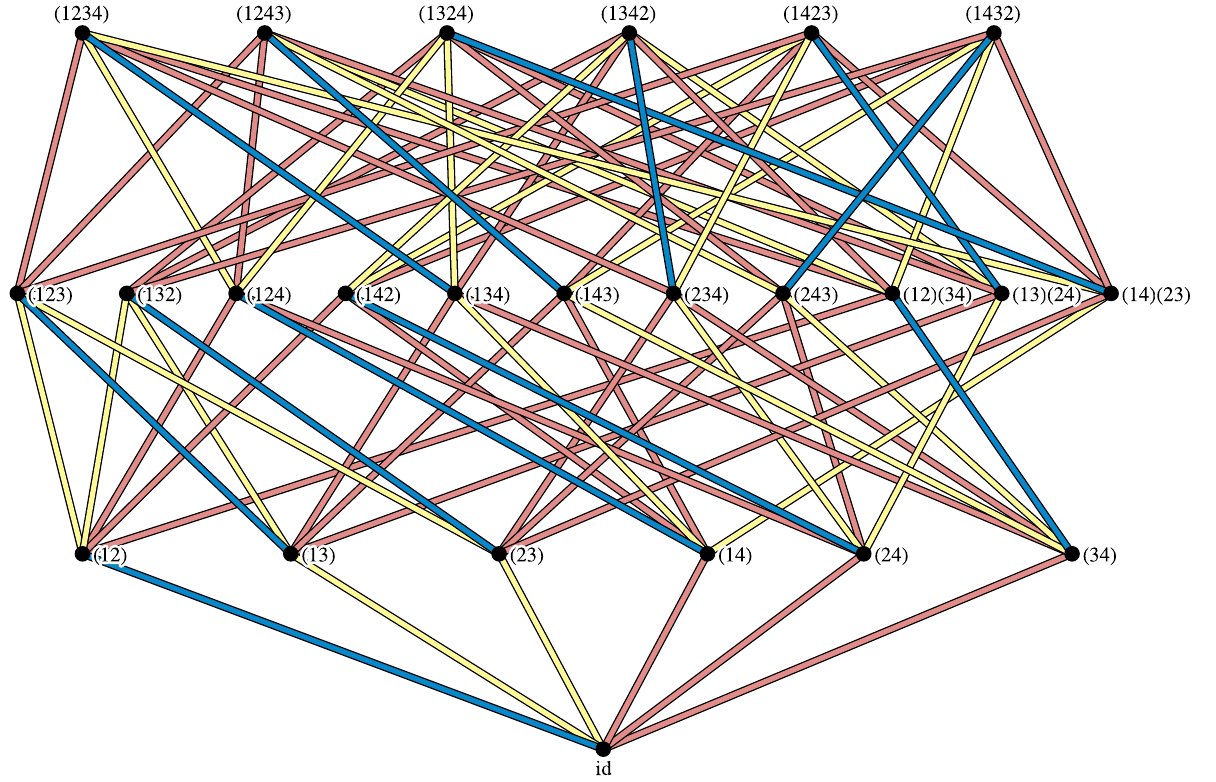}
	\caption{\label{fig:Cayley} $S(4)$ with the Biane-Stanley
	edge-labelling.}
\end{figure}

\subsection{}
The number $\hur^r(\alpha,\beta)$, which counts walks 
as above, but without the monotonicity constraint, is
a \emph{double Hurwitz number}.  The double Hurwitz
numbers are important quantities in classical 
and modern enumerative geometry, see 
\cite{GJV,Okounkov}.
Reversing a classical construction due to
Hurwitz \cite{Hurwitz}, we have that

    \begin{equation*}
    \frac{1}{d!} \hur^r(\alpha,\beta) =
    \sum_{(X,f)} \frac{1}{|\operatorname{Aut}(X,f)|},
    \end{equation*}

\noindent
where the sum runs over all isomorphism
classes of pairs $(X,f)$ in which $X$ is
a compact, connected Riemann surface and $f:X \rightarrow 
\mathbf{P}^1$ is a degree $d$ mapping 
to the Riemann sphere with profile 
$\alpha$ over $\infty$, profile $\beta$
over $0$, and simple ramification over 
the $r$th roots of unity.
By the Riemann-Hurwitz formula, such a branched
covering exists if and only if

    \begin{equation*}
    \label{eqn:RiemannHurwitz}
        g = \frac{r+2-\ell(\alpha)-\ell(\beta)}{2}
    \end{equation*}
    
\noindent
is a non-negative integer, in which
case $g$ is the genus of $X$.  Here
$\ell(\alpha)$ is the number of parts
in the partition $\alpha \vdash d$, and
likewise for $\ell(\beta)$.
We write 
$\hur^r(\alpha,\beta) = \hur_g(\alpha,\beta)$,
with the understanding that $r$
and $g$ determine one another via 
Riemann-Hurwitz.

\subsection{}
Following the terminology of \cite{GGN3},
we refer to the numbers 
$\mon^r(\alpha,\beta)=\mon_g(\alpha,\beta)$ as the
\emph{monotone} double Hurwitz numbers.
The expansion in Theorem
\ref{thm:LeadingDerivatives} may
equivalently be written

    \begin{equation}
    \label{eqn:LeadingDerivatives}
        C_N(\alpha,\beta) =
        (-1)^{\ell(\alpha) + \ell(\beta)}
        N^{2-d-\ell(\alpha)-\ell(\beta)}
        \sum_{g \geq 0} \frac{\mon_g(\alpha,\beta)}{N^{2g}}.
    \end{equation}
    
\noindent
This expansion renders the asymptotics 
of the HCIZ integral transparent in 
virtually any scaling regime.  In particular,
one obtains the following limits.

    \begin{prop}
    \label{prop:FreeCumulant}
        Under the hypotheses of Theorem
        \ref{thm:Main},
        for any fixed $d \in \mathbb{N}$
        and $a_1,\dots,a_k \in \mathbb{C}$, we
        have
        
            \begin{equation*}
            \begin{split}
                &\lim_{N \rightarrow \infty}\frac{1}{N}
                \sum_{\alpha,\beta \vdash d} C_N(\alpha,\beta)
                p_{\alpha}(a_1,\dots,a_k)
                p_{\beta}(b_1^{(N)},\dots,b_N^{(N)}) \\
                =&\ p_d(a_1,\dots,a_k)
                \sum_{\beta \vdash d} (-1)^{1+\ell(\beta)} 
                \mon_0(d,\beta)\psi_{\beta},
            \end{split}
            \end{equation*}
            
        \noindent
        where $\psi_\beta=\prod_i \psi_{\beta_i}$.
    \end{prop}
    
    \begin{proof}
        According to \eqref{eqn:LeadingDerivatives},
        we have 
        
        \begin{equation*}
        \begin{split}
                &\frac{1}{N}
                \sum_{\alpha,\beta \vdash d} C_N(\alpha,\beta)
                p_{\alpha}(a_1,\dots,a_k)
                p_{\beta}(b_1^{(N)},\dots,b_N^{(N)}) \\
                =& \sum_{\alpha \vdash d} (-1)^{\ell(\alpha)}
                \frac{p_{\alpha}(a_1,\dots,a_k)}{N^{\ell(\alpha)-1}}
                \sum_{\beta \vdash d} (-1)^{\ell(\beta)}
                \frac{
                p_{\beta}(\frac{b_1^{(N)}}{N},\dots,\frac{b_N^{(N)}}{N})}
                {N^{\ell(\beta)}}
                \sum_{g=0}^{\infty} \frac{\mon_g(\alpha,\beta)}{N^{2g}}
        \end{split}
        \end{equation*}
        
        \noindent
        for any $N \geq d$.
        From the definition of $\mon_g(\alpha,\beta)$, 
        we have the upper bound
        
        \begin{equation*}
            \mon_g(\alpha,\beta) \leq (d!)^{2g + 
            \ell(\alpha)+\ell(\beta)} \leq
            (d!)^{2g+2d}.
        \end{equation*}
        
        \noindent
        Thus
        
        \begin{equation*}
            \sum_{g=0}^{\infty} \frac{\mon_g(\alpha,\beta)}{N^{2g}}
            = \mon_0(\alpha,\beta) +O\left(\frac{1}{N^2}\right)
        \end{equation*}

        \noindent
        as $N \rightarrow \infty$, uniformly in $\alpha,\beta$.
        
        The weak convergence of $\nu^{(N)}$ to $\nu$,
        the measure on $[-M,M]$ with moments
        $\{\psi_m : m \in \mathbb{N}\}$, 
        is equivalent to the limits
        
        \begin{equation*}
            \lim_{N \rightarrow \infty}
            \frac{p_m(\frac{b_1^{(N)}}{N},\dots,\frac{b_N^{(N)}}{N})}{N}
            =\psi_m, \quad m \in \mathbb{N}.
        \end{equation*}
    \end{proof}

\subsection{}
The numbers $\mon_g(d,\beta)$ are \emph{one-part}
monotone double Hurwitz numbers; their classical
counterparts $\hur_g(d,\beta)$ are analyzed in \cite{GJV}.
The sum

    \begin{equation*}
        K_d = \sum_{\beta \vdash d} (-1)^{1+\ell(\beta)} 
        \mon_0(d,\beta)\psi_{\beta}    
    \end{equation*}
    
\noindent
which emerges in Proposition \ref{prop:FreeCumulant}
is an element of $\mathbb{Z}[\psi_1,\dots,\psi_d]$,
homogeneous of degree $d$ with respect to the grading
$\deg(\psi_m)=m$.  In fact, $K_d$
is, up to a simple factor, the $d$th \emph{free cumulant}
$\kappa_d$ of the measure $\nu$:

    \begin{equation}
    \label{eqn:FreeCumulant}
        K_d = (d-1)! \kappa_d.
    \end{equation}

We recall that the free cumulants of a probability
measure are obtained by replacing the lattice of all
partitions with the lattice of noncrossing partitions
in the moment-cumulant formula, see e.g. \cite{NS}.
The identity \eqref{eqn:FreeCumulant} may be
established in a purely combinatorial way, by
viewing the noncrossing partition lattice
$NC(d)$ as the set of geodesic paths 
$(1)\dots (d) \rightarrow (1\ \dots\ d)$
on the Cayley graph of $S(d)$ and using the
Kreweras antiautomorphism.  For our purposes,
we only require explicit knowledge of $K_1$
and $K_2$, which can be computed directly
from the definition of the monotone
double Hurwitz numbers:

    \begin{equation*}
    \begin{split}
        &\mon_0(1,1)=1 \implies K_1=\psi_1 \\
        &\mon_0(2,2) = \mon_0(2,11) = 1 \implies
        K_2=\psi_2-\psi_1^2.
    \end{split}
    \end{equation*}

\noindent
We thus leave the proof of
\eqref{eqn:FreeCumulant}
to the interested reader.  

\subsection{}
The absolute summability of the series

    \begin{equation*}
        \sum_{d=1}^{\infty} \frac{z^d}{d!} K_d
    \end{equation*}
    
\noindent
follows from \cite[Theorem 3.4]{GGN3}.
Arguing as in \cite[Theorem 4.1]{GGN3},
Proposition \ref{prop:FreeCumulant}
may be promoted to the following
scaling limit of the HCIZ integral,
which is closely related to the 
results of \cite{CS,GM}.

    \begin{prop}
    \label{prop:FPlimit}
        Let $k \in \mathbb{N}$ be fixed.
        Under the assumptions of Theorem
        \ref{thm:Main}, there exists $\varepsilon > 0$
        such that
        
        \begin{equation*}
        \frac{1}{N}\log
        \int\limits_{U(N)} e^{z\Tr\diag(a_1,\dots,a_k,0,\dots,0)
        U\diag(b_1^{(N)},\dots,b_N^{(N)})U^{-1}}
        \mathrm{d}U 
        \rightarrow
        \sum_{d=1}^\infty \frac{z^d}{d!} p_d(a_1,\dots,a_k)K_d,
        \end{equation*}
        
        \noindent
        uniformly on compact subsets of
        $\{(z;a_1,\dots,a_k) \in \mathbb{C} \times \mathbb{C}^k
        : |za_i| < \varepsilon \}$.
    \end{prop}

    Tuning the coupling constant to $z=N^{-1/2}$,
    Proposition \ref{prop:FPlimit} yields the following corollary.
    
    \begin{cor}
    \label{cor:TilingLimit}
        Let $k \in \mathbb{N}$ be fixed.
        Under the assumptions of Theorem
        \ref{thm:Main}, we have
        the $N \rightarrow \infty$
        asymptotic expansion
        
        \begin{equation*}
        \log
        \int\limits_{U(N)} e^{\frac{1}{\sqrt{N}}
        \Tr\diag(a_1,\dots,a_k,0,\dots,0)
        U\diag(b_1^{(N)},\dots,b_N^{(N)})U^{-1}}
        \mathrm{d}U \sim
        \sum_{d=1}^\infty 
        \frac{K_d}{d!} p_d(a_1,\dots,a_k)N^{1-\frac{d}{2}},
        \end{equation*}
        
        \noindent
        uniformly on compact subsets
        of $\mathbb{C}^k$.
    \end{cor}

    \subsection{}
    Combining Corollary \ref{cor:TilingLimit}
    with the fact that
    
    \begin{equation*}
        \log \frac{a}{e^a-1}=
        -\log \frac{e^a-1}{a} = -\frac{1}{2} \frac{a^1}{1!}
        - \frac{1}{12} \frac{a^2}{2!} + \dots
    \end{equation*}
    
    \noindent 
    is negative one times the generating
    function for the classical cumulants
    $c_1,c_2,\dots$ of uniform measure on
    $[0,1]$, Proposition \ref{prop:Key} 
    yields the asymptotic expansion
    
    \begin{equation*}
        \log L_k^{(N)}(\frac{a_1}{\sqrt{N}},\dots,
        \frac{a_k}{\sqrt{N}}) \sim \sum_{d=1}^\infty 
        \frac{K_d-c_d}{d!} p_d(a_1,\dots,a_k)N^{1-\frac{d}{2}},
    \end{equation*}

    \noindent
    uniformly on compact subsets of $\mathbb{C}^k$.
    In particular,
    
    \begin{equation*}
        \log L_k^{(N)}(\frac{a_1}{\sqrt{N}},\dots,
        \frac{a_k}{\sqrt{N}})
        =\sqrt{N}(\psi_1-\frac{1}{2})p_1(a_1,\dots,a_k)
        +\frac{1}{2}(\psi_2-\psi_1^2-\frac{1}{12})p_2(a_1,\dots,a_k)
        + O\left( \frac{1}{\sqrt{N}} \right)
    \end{equation*}

    \noindent
    as $N \rightarrow \infty$.  
    Since a $k \times k$ standard GUE random 
    matrix $X_k$ is characterized by the
    log-Laplace transform
    
    \begin{equation*}
        \log \mathbf{E}[e^{\Tr AX_k}]=
        \frac{1}{2}\Tr A^2,
    \end{equation*}
    
    \noindent
    and since $H(k)$ is a finite-dimensional
    Euclidean space with the inner product
    $\langle A, B\rangle = \Tr AB$, Theorem
    \ref{thm:Main} follows from the above
    quadratic approximation and the L\'evy
    continuity theorem.

\bibliographystyle{amsplain}

\end{document}